\documentclass[a4paper,10pt]{amsart}

\usepackage[T1]{fontenc}
\usepackage[utf8x]{inputenc}
\usepackage[english]{babel}
\usepackage{xspace}
\usepackage{hyperref}
\usepackage{subfigure}
\usepackage{textcomp}

\usepackage{mathrsfs}
\usepackage{amsfonts}
\usepackage{amssymb}
\usepackage{amsmath}

\usepackage{mathbbol}
\usepackage{color}
\usepackage{graphicx}
\DeclareMathAlphabet{\mathpzc}{OT1}{pzc}{m}{it}

\newcommand{\macro}[3]{\newcommand{#1}[#3]{#2}}

\newcommand\macros[4]
                   {%
                     \newenvironment{#1}[1][#2]{\par\vspace{1ex}
                       #3 \hspace{0.5em}#4}
                                    {\nopagebreak%
                                      \strut\nopagebreak%
                                      \par\vspace{2ex}
                                    }
                   }

\macro{\twd}{\operatorname{tw}(#1)}{1}
\macro{\pw}{\operatorname{pw}(#1)}{1}
\macro{\bwd}{\operatorname{bw}(#1)}{1}
\macro{\cwd}{\operatorname{cw}(#1)}{1}
\macro{\rwd}{\operatorname{rw}(#1)}{1}
\macro{\sn}{\operatorname{sn}(#1)}{1}
\macro{\bn}{\operatorname{bn}(#1)}{1}
\macro{\Ht}{\operatorname{h}(#1)}{1}	

\macro{\brw}{\operatorname{brw}(#1)}{1}
\macro{\Frw}{{#1}\textrm{-}\operatorname{rwd}(#2)}{2}
\macro{\Qrw}{\Frwd{\gfq}{#1}}{1}
\macro{\frw}{\Frwd{\bF}{#1}}{1}
\macro{\Fbrw}{{#1}\textrm{-}\operatorname{brwd}(#2)}{2}
\macro{\fbrw}{\Fbrwd{\bF}{#1}}{1}

\macro{\lrw}{\operatorname{lrw}(#1)}{1}
\macro{\lbw}{\operatorname{lbw}(#1)}{1}
\macro{\lcw}{\operatorname{lcw}(#1)}{1}

\macro{\const}{\mathbf{#1}}{1}
\macro{\angl}{\mathop\langle #1 \mathop\rangle}{1}
\macro{\up}{\ulcorner #1\urcorner}{1}
\macro{\card}{\left|{#1}\right|}{1}
\macro{\floor}{\left\lfloor{#1}\right\rfloor}{1}
\macro{\ceil}{\left\lceil{#1}\right\rceil}{1}
\macro{\pare}{\left({#1}\right)}{1}
\macro{\crochet}{\left[{#1}\right]}{1}
\macro{\set}{\left\{{#1}\right\}}{1}
\macro{\range}{\set{{#1},\ldots,{#2}}}{2}
\macro{\mat}{M_{#1}}{1}
\macro{\matind}{{#1}[{#2},{#3}]}{3}
\macro{\matgind}{\matind{\matg}{#1}{#2}}{2}
\macro{\leaves}{\operatorname{L}_{#1}}{1}
\macro{\cutrk}{{#1}\textrm{-}\operatorname{cutrk}}{1}
\macro{\comp}{(X^{#1},#2\backslash X^{#1})}{2}
\macro{\supp}{\mathpzc{u}(#1)}{1}
\macro{\field}{\mathbb{F}_{#1}}{1}
\macro{\subg}{#1\textrm{-}#2}{2}
\macro{\bicutrk}{{#1}\textrm{-}\operatorname{bicutrk}}{1}

\def\gfq{\operatorname{\field{4}}}
\def\rk{\operatorname{rk}}

\def\matg{\mat{G}}

\def\ucutrk{\operatorname{cutrk}}

\def\bF{\mathbb{F}}

\def\cD{\mathcal{D}}

\macros{defi}{Example}{\noindent{\bf #1}}{}
\macros{mthm}{Main Theorem}{\noindent{\bf #1}}{\itshape}

\usepackage{calc}

{\end{list}}
\newtheorem{thm}{Theorem}

\newtheorem{cor}{Corollary}
\newtheorem{prop}{Proposition}

\title{A Note on Graphs of Linear Rank-Width $1$}
\date{\today}
\author[B. Bui-Xuan]{Binh-Minh
  Bui-Xuan} 
\address{Universit{\'e} Paris 6, LIP6, CNRS, France}
\author[M.M. Kant\'e \and V. Limouzy]{Mamadou Moustapha Kant\'e \and
  Vincent Limouzy}
\address{Clermont-Universit{\'e}, Universit{\'e} Blaise Pascal, LIMOS,
  CNRS, France}
\thanks{M.M. Kant\'e and V. Limouzy are supported by
    the French Agency for Research under the 
    DORSO project (2011-2015)}

\begin{document}
\begin{abstract}
We prove that a connected graph has linear rank-width $1$ if and only
if it is a distance-hereditary graph and its split decomposition tree
is a path. An immediate consequence is that one can decide in linear
time whether a graph has linear rank-width at most $1$, and give an
obstruction if not. Other immediate consequences are several
characterisations of graphs of linear rank-width $1$. In particular a
connected graph has linear rank-width $1$ if and only if it is locally
equivalent to a caterpillar if and only if it is a vertex-minor of a
path [O-joung Kwon and Sang-il Oum, Graphs of small rank-width are
  pivot-minors of graphs of small tree-width, to appear in Discrete
  Applied Mathematics] if and only if it does not contain the co-$K_2$
graph, the Net graph and the $5$-cycle graph as vertex-minors [Isolde
  Adler, Arthur M. Farley and Andrzej Proskurowski, Obstructions for
  linear rank-width at most $1$, to appear in Discrete Applied
  Mathematics].
\end{abstract}
\maketitle

\section{Introduction}

In their investigations for the recognition of graphs of bounded
\emph{clique-width} \cite{CourcelleO00} Oum and Seymour introduced the
notion of \emph{rank-width} \cite{OumS06} of a graph. Rank-width
appeared to have several nice combinatorial properties, in particular
it is related to the \emph{vertex-minor} inclusion, and have proven
the last years its importance in studying the structure of graphs of
bounded clique-width \cite{Oum05a,OumS06,CourcelleO07,HlinenyO08}. \emph{Linear
  rank-width} is related to rank-width in the same way
\emph{path-width} \cite{RobertsonS83} is related to \emph{tree-width}
\cite{RobertsonS86}.  Indeed, linear rank-width is the linearised
version of rank-width and studying graphs of bounded linear rank-width
is a first step in studying the structure of graphs of bounded
rank-width which is not yet well understood. Not much is known about
linear rank-width. The computation of the linear rank-width of forests
is investigated in \cite{AdlerK13} and it is proved in \cite{KwonO13}
that graphs of linear rank-width $k$ are vertex-minors of graphs of
path-width at most $k+1$. Ganian defined in \cite{Ganian10} the notion
of \emph{thread graphs} and proved that they correspond exactly to
graphs of linear rank-width $1$ and authors of \cite{AdlerFP13} used
it to exhibit the set of vertex-minor obstructions for linear
rank-width $1$.  In this paper we investigate in a different way the
structure of graphs of linear rank-width $1$.

Distance hereditary graphs \cite{BandeltM86} are a well-known and
well-studied class of graphs because of their multiple nice
algorithmic properties. They admit several characterisations, in
particular they correspond exactly to graphs of rank-width at most $1$
\cite{Oum05a} and are the graphs that are totally decomposable with
respect to \emph{split decomposition} \cite{CunninghamE80}. Split
decomposition is a graph decomposition introduced by Cunningham and
Edmonds and has proved its importance in algorithmic and structural
graph theory (see for instance
\cite{Bouchet87,Bouchet88,Bouchet94,Courcelle06,Rao08} to cite a
few). We give in this paper the following characterisation of graphs
of linear rank-width $1$ which implies all the known characterisations
of graphs of linear rank-width $1$. 

\begin{thm}
\label{thm:lrw1}
 A  connected graph $G$ has linear rank width $1$ if and only
 if it is a distance hereditary graph and its split decomposition tree
 is a path.
\end{thm}

A first consequence of this theorem is that we can derive in a more
direct way than in \cite{AdlerFP13} the set of induced subgraph (or
vertex-minor of pivot-minor) obstructions for linear rank-width
$1$. Another consequence is a simple linear time algorithm for
recognising graphs of linear rank-width $1$ (the only known one prior
to this algorithm is the one that uses logical tools
\cite{CourcelleO07} and is not really practical). Our algorithm gives
moreover an obstruction if it exists. Notice that a polynomial time
algorithm for recognising graphs of linear rank-width $1$ is a
consequence of the characterisation in terms of \emph{thread graphs} given in
\cite{Ganian10}.

The paper is organised as follows. Some definitions and notations are
given in Section \ref{sec:2}. In Section \ref{sec:3} we introduce the
notion of split decomposition and prove our main theorem. We derive
several characterisations and give a simple linear time algorithm
(with a certificate) for the recognition of graphs of linear
rank-width at most $1$.

%% Parler de: 
%% \begin{enumerate}
%%  \item distance hereditaire = rw1 = totalement dec par split.
%%   \item Que signifie totalement decomp et caterpillar.
  
%% \end{enumerate}

\section{Preliminaries}\label{sec:2}

%% If
%% $f:A\to B$ is an application, we let $f_{\mid X}$, for $X\subseteq A$,
%% be the application $f_{\mid X}:X\to B$ where for every $a\in
%% X,\ f_{\mid X}(a) = f(a)$.

For two sets $A$ and $B$, we let $A\backslash B$ be the set $\{x\in
A\mid x\notin B\}$.  We often write $x$ to denote the set $\{x\}$.
For sets $R$ and $C$, an \emph{$(R,C)$-matrix} is a matrix where the
rows are indexed by elements in $R$ and columns indexed by elements in
$C$. For an $(R,C)$-matrix $M$, if $X\subseteq R$ and $Y\subseteq C$,
we let $\matind{M}{X}{Y}$ be the sub-matrix of $M$ the rows and the
columns of which are indexed by $X$ and $Y$ respectively. We let $\rk$
be the matrix rank-function (the field will be clear from the
context). %%We denote by $M^t$ the transpose of a matrix $M$.

Our graph terminology is standard, see for instance
\cite{Diestel05}. A graph $G$ is a pair $(V(G),E(G))$ where $V(G)$ is
the set of vertices and $E(G)$, the set of edges, is a set of
unordered pairs of $V(G)$.  An edge between $x$ and $y$ in a graph is
denoted by $xy$ (equivalently $yx$). The subgraph of a graph $G$
induced by $X\subseteq V(G)$ is denoted by $G[X]$. Two graphs $G$ and
$H$ are \emph{isomorphic} if there exists a bijection $\varphi:V(G)\to
V(H)$ such that $xy\in E(G)$ if and only if $\varphi(x)\varphi(y)\in
E(H)$. All graphs are finite and loop-free.

A \emph{tree} is an acyclic connected graph. In order to avoid
confusions in some lemmas, we will call \emph{nodes} the vertices of
trees. The nodes of degree $1$ are called \emph{leaves}.  A
\emph{path} is a tree the vertices of which have all degree $2$,
except two that have degree $1$. A \emph{caterpillar} is a tree such
that the removal of leaves results in a path.

A graph is \emph{distance hereditary} if its induced subpaths are
isometric \cite{BandeltM86}. Examples of distance hereditary graphs
are trees, cliques, etc. There exist several characterisations of
distance hereditary graphs. See for instance \cite[Theorem
  1]{KanteR09} for a summary of some known characterisations of
distance hereditary graphs.

The adjacency matrix of a graph $G$ is the $(V(G),V(G))$-matrix $A_G$
over $GF(2)$ where $A_G[x,y]=1$ if and only if $xy\in E(G)$. For a
graph $G$, let $x_1,\ldots,x_n$ be a linear ordering of its
vertices. For each index $i$, we let $X_i:=\{x_1,\ldots, x_i\}$ and
$\overline{X_i}:=\{x_{i+1},\ldots,x_n\}$. The \emph{cutrank} of the
ordering $x_1,\ldots,x_n$ is defined as 
\begin{align*}
  \ucutrk_G(x_1,\ldots,x_n) &= \max
  \{\rk\left(A_G[X_i,\overline{X_i}]\right) \mid 1\leq i \leq n\}.
\end{align*}

The \emph{linear rank-width} of a graph $G$ is defined as 
\begin{align*}
  \lrw{G} &= \min \{\ucutrk_G(x_1,\ldots,x_n)\mid
  x_1,\ldots,x_n\ \textrm{is a linear ordering of $V(G)$}\}.
\end{align*}

The linear ordering $c,d,a,b,g,e,f$ of the graph in Figure
\ref{fig:ex-split} has cutrank $1$. It is worth noticing that if a
graph $G$ is not connected and its connected components are
$C_1,\ldots, C_t$, then $\lrw{G} = \max\{\lrw{C_i}\}_{1\leq i\leq t}$
since it suffices to concatenate in any order the linear ordering of
optimal cutrank of its connected components.

For a graph $G$ and a vertex $x$ of $G$, the \emph{local
  complementation at $x$} consists in replacing the subgraph induced
on the neighbours of $x$ by its complement. The resulting graph is
denoted by $G*x$. A graph $H$ is \emph{locally equivalent} to a graph
$G$ if $H$ is obtained from $G$ by applying a sequence of local
complementations, and $H$ is called a \emph{vertex-minor} of $G$ if
$H$ is isomorphic to an induced subgraph of a graph locally equivalent
to $G$.  The following relates vertex-minor and linear rank-width.

\begin{prop}[\cite{Oum05a}]\label{prop:2.1} Let $G$ and $H$ be two
  graphs. If $H$ is locally equivalent to $G$, then
  $\lrw{H}=\lrw{G}$. If $H$ is a vertex-minor of $G$, then $\lrw{H}
  \leq \lrw{G}$.
\end{prop}

\section{Linear rank-width $1$} \label{sec:3}

We prove in this section our main theorem.  Let us make precise some
terminologies and notations.

\subsection{Split decomposition}

Two bipartitions $\{X_1,X_2\}$ and $\{Y_1,Y_2\}$ of a set $V$
\emph{overlap} if $X_i\cap Y_j\ne \emptyset$ for all $i,j\in \{1,2\}$.
A \emph{split} in a connected graph $G$ is a bipartition $\{X,Y\}$ of
the vertex set $V(G)$ such that $|X|,|Y|\geq 2$ and
$\rk\left(A_G[X,Y]\right) = 1$. A split $\{X,Y\}$ is \emph{strong} if
there is no other split $\{X',Y'\}$ such that $\{X,Y\}$ and
$\{X',Y'\}$ overlap. Figure \ref{fig:3.1.1} shows a schematic view of
splits.  Notice that not all graphs have a split and those without a
split are called \emph{prime}. We follow \cite{Courcelle06} for the
definition of a \emph{split decomposition tree}.

\begin{figure}%[h!]
	\begin{center}
	  {\includegraphics{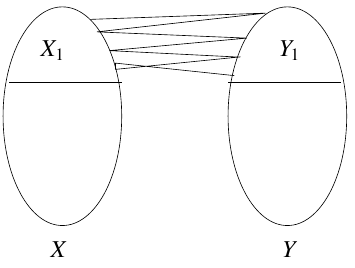}}
	\caption{Schematic view of splits.}
	\label{fig:3.1.1}
	\end{center}
\end{figure}

If $\{X,Y\}$ is a split, then we let $G^X$ and $G^Y$ be the graphs
with vertex set $X\cup \{h_X\}$ and $Y\cup \{h_Y\}$ respectively where the
vertices $h_X$ and $h_Y$ are new and called \emph{neighbour
  markers} of $G^X$ and $G^Y$ respectively, and with edge set 
\begin{align*}
E(G^X) &:=E(G[X])\cup \{xh_X\mid x\in
X\ \textrm{and}\ N_G(x)\cap Y \ne \emptyset\},\qquad \textrm{and}\\ 
E(G^Y)&:=E(G[Y])\cup
\{yh_Y\mid y\in Y\  \textrm{and}\ N_G(y) \cap X\ne \emptyset\}.
\end{align*}

A \emph{decomposition} of a connected graph $G$ is defined inductively
as follows: $\{G\}$ is the only decomposition of size $1$. If
$\{G_1,\ldots, G_n\}$ is a decomposition of size $n$ of $G$, then if
$G_i$ has a split $\{X,Y\}$, then $\{G_1,\ldots, G_{i-1},
G^X, G^Y, G_{i+1}, \ldots, G_n\}$ is a
decomposition of size $n+1$. Notice that the decomposition process
must terminate because the new graphs $G^X$ and
$G^Y$ are smaller than $G_i$. The graphs $G_i$ of a
decomposition are called \emph{blocks}. If two blocks have neighbour
markers, we call them \emph{neighbour blocks}. 

For every decomposition $\cD$ of a
connected graph $G$ we associate the graph $S(\cD)$ with vertex set
$\bigcup_{G_i\in \cD} V(G_i)$ and edge set 
\begin{align*}
  %% V(S(\cD)) & :=\bigcup_{G_i\in \cD} V(G_i),\\
  %% E(S(\cD))
  &  \left(\bigcup_{G_i\in \cD} E(G_i)\right) \cup \{h_Xh_Y\mid
  h_X,h_Y\ \textrm{are neighbour markers}\}.
\end{align*}

Edges in $E(S(\cD))$ between neighbour markers of $\cD$ are called
\emph{marked edges} and the others are called \emph{solid edges}. One
notices that subgraphs of $S(\cD)$ induced by solid edges are blocks
of $\cD$. Observe that each mark edge is an isthmus, and
the marked edges form a matching. 

Two decompositions $\cD_1$ and $\cD_2$ of a connected graph $G$ are
\emph{isomorphic} if there exists a graph isomorphism $f$ between
$S(\cD_1)$ and $S(\cD_2)$ which preserves the marked edges, and such
that $f(x)=x$ for all $x\in V(G)$. It is worth noticing that a graph
can have several non isomorphic decompositions. However, a
\emph{canonical decomposition} can be defined. A decomposition is
\emph{canonical} if and only if: (i) each block is either prime
(called \emph{prime block}), or is isomorphic to a clique of size at
least $3$ (called \emph{clique block}) or to a star of size at least
$3$ (called \emph{star block}), (ii) no two clique blocks are
neighbour, and (iii) if two star blocks are neighbour, then either
their markers are both centres or both not centres.  The following
theorem is due to Cunningham and Edmonds \cite{CunninghamE80}, and
Dahlhaus \cite{Dahlhaus00}.

\begin{thm}[\cite{CunninghamE80,Dahlhaus00}] \label{thm:CED} Every
  connected graph has a unique canonical decomposition, up
  to isomorphism. It can be obtained by iterated splitting relative to
  strong splits. This canonical decomposition can be computed in time
  $O(n+m)$ for every graph $G$ with $n$ vertices and $m$ edges.
\end{thm}

The canonical decomposition of a connected graph $G$ constructed in
Theorem \ref{thm:CED} is called \emph{split decomposition} and we will
denote it by $\cD_G$.  Since marked edges of $\cD_G$ are isthmus and
form a matching, if we contract the solid edges in $S(\cD_G)$, we
obtain a tree called \emph{split decomposition tree} of $G$ and
denoted by $T_G$. For every node $u$ of $T_G$, we denote by $b_G(u)$
the block of $\cD_G$ the edges of which are contracted to get $u$, and
we let $V(u)$ be $V\left(b_G(u)\right)\cap V(G)$.  For an edge $uv$ of
$T_G$, we denote by $G^{uv}$ the subgraph of $G$ induced by
$\bigcup\limits_{w\in T_G^{uv}} V(u)$ where $T_G^{uv}$ is the subtree
of $T_G$ induced by those nodes $w$ of $T_G$ such that a path from $u$
to $w$ does not contain $v$. Notice that for every edge $uv$ of $T_G$,
$\{V(G^{uv}),V(G)\setminus{V(G^{uv})}\}$ is a strong split of $G$. We
finish these preliminaries with the following characterisation of
distance hereditary graphs. 

%% \begin{thm}[\cite{Bouchet88}]\label{thm:Bouchet88} A connected 
%%     graph is a distance-hereditary graph if and only if each block of
%%     its split decomposition is either a clique or a star block.
%% \end{thm}

%% As a corollary, we have the following characterisation of
%% distance-hereditary graphs.

\begin{thm}[\cite{CunninghamE80}]\label{thm:CE80} A connected graph $G$ is
  distance hereditary if and only if for every node $u$ of $T_G$ and
  for every $\emptyset \subsetneq W \subsetneq N_{T_G}(u)$, the
  bipartition $\{X, V(G)\setminus {X}\}$ with
  $X:=\left(\bigcup\limits_{v\in W} V(G^{vu})\right) \cup X'$ such
  that $X'\subseteq V(u)$ is a split in $G$, provided that
  $|X|,|V(G)\setminus {X}| \geq 2$.
\end{thm}

\begin{figure}%[h!]
	\begin{center}
		\includegraphics[scale=1.0]{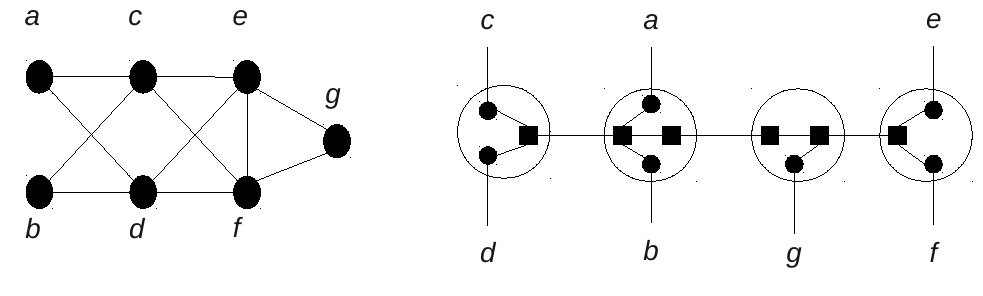}
	\caption{An example of a graph and its split decomposition tree.}
	\label{fig:ex-split}
	\end{center}
\end{figure}

\subsection{Charaterizing graphs of linear rank-width $1$} It is
folklore to verify that the rank-width of a graph is smaller (or
equal) than its linear rank-width. Hence, a connected graph of linear
rank-width $1$ has necessarily rank-width $1$. It is proved in
\cite{Oum05a} that a connected graph has rank-width $1$ if and only if
it is a \emph{distance hereditary} graph.  Therefore, Theorem
\ref{thm:lrw1} follows from Propositions \ref{prop:sdtlrw1} and
\ref{prop:lrw1sdt} below.

\begin{prop}
\label{prop:sdtlrw1}
 Let $G$ be a connected distance hereditary graph such that $T_G$ is a
 path. Then $\lrw{G}=1$.
\end{prop}

\begin{proof} We will show how to turn $T_G$ into a linear ordering of
  $V(G)$ of cutrank $1$.  Let us enumerate the nodes of $T_G$ as $u_1,
  \ldots, u_p$ from left to right. For every $1\leq i \leq p$, let
  $\pi_i$ be any linear ordering of $V(u_i)$, and let $\pi:=\pi_1\cdot
  \cdots \cdot \pi_p$ be the concatenation of the orderings
  $\pi_1,\ldots, \pi_p$.  Since $\{V(u_1), \ldots, V(u_p)\}$ is a
  partition of $V(G)$, $\pi$ is clearly a linear ordering of
  $V(G)$. We claim that its cutrank is $1$. Indeed let $i$ be an index
  of this ordering and let $X_i:=\{x_1,\ldots,x_i\}$. Assume without
  loss of generality that $|X_i|,|\overline{X_i}|\geq 2$, otherwise
  we have trivially $\rk\left(A_G[X_i,\overline{X_i}]\right)=1$. Since $T_G$
  is a path, then $X_i$ is equal to $\bigcup_{1\leq l \leq j-1} V(u_l) \cup
  X'$ with $X'\subseteq V(u_{j})$ for some $1\leq j < p$. By Theorem
  \ref{thm:CE80}, $\{X_i,\overline{X_i}\}$ is a split in $G$, and
  hence $\rk\left(A_G[X_i,\overline{X_i}]\right)=1$.
\end{proof}

The next proposition gives the converse direction of Proposition
\ref{prop:sdtlrw1}. 

\begin{prop}
\label{prop:lrw1sdt}
 Let $G$ be a connected graph of linear rank-width $1$. Then $G$ is distance
 hereditary and $T_G$ is a path.
\end{prop}

\begin{proof} Let $G$ be a graph with $\lrw{G}=1$. Hence, $G$ is
  distance hereditary (see the paragraph before Proposition
  \ref{prop:sdtlrw1}). Let $\pi:=x_1,\ldots,x_n$ be a linear ordering
  of $V(G)$ of cutrank $1$. It suffices to prove that every strong
  split of $G$ is of the form $\{X_i,\overline{X_i}\}$ for some index
  $1< i < n$. Suppose this is not the case and let $\{X,Y\}$ be a
  strong split of $G$ with $\{X,Y\} \ne \{X_i,\overline{X_i}\}$ for
  every $1<i<n$. Without loss of generality we can assume that $x_1\in
  X$ and let $j$ be the smallest index such that $x_j\notin X$. First
  of all notice that $2\leq j \leq n-1$ otherwise $|Y|\leq 1$ and then
  $\{X,Y\}$ would not be a split. Therefore, $\{X_j,\overline{X_j}\}$
  is a split because $\pi$ is a linear ordering of $V(G)$ of cutrank
  $1$. We have that
  \begin{itemize}
  \item $x_1\in X\cap X_j$ and $x_j\in Y\cap X_j$,
  \item $X\cap \overline{X_j} \ne \emptyset$ otherwise $X$ would equal
    $X_{j-1}$,
  \item $Y\cap \overline{X_j}\ne \emptyset$ because otherwise
    $|Y|=1$. 
  \end{itemize} 
 Therefore, $\{X,Y\}$ and $\{X_j,\overline{X_j}\}$ overlap, which
 contradicts the fact that $\{X,Y\}$ is a strong split. 
\end{proof}

Figure \ref{fig:lrw2} shows a graph of linear rank-width exactly $2$. 

\begin{figure}[h!]
	\begin{center}
		\includegraphics[scale=1.0]{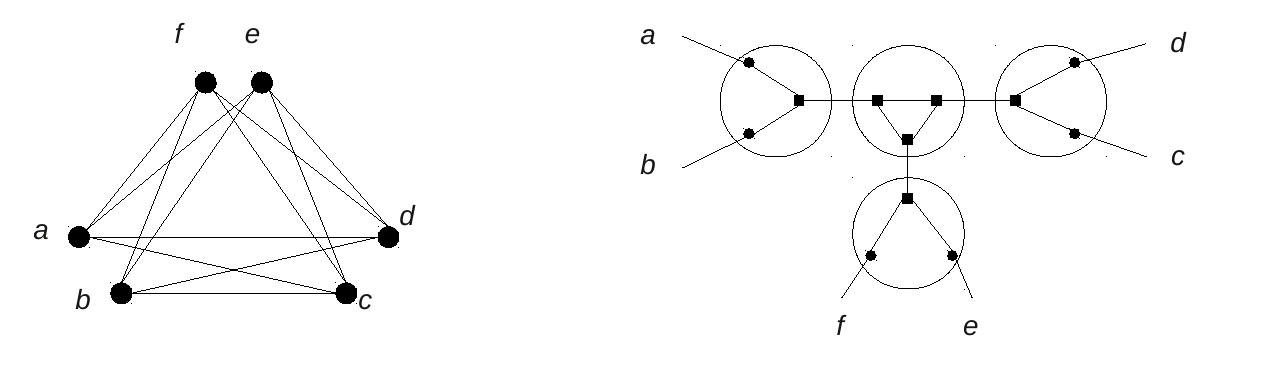}
	\caption{A graph of linear rank-width 2}
	\label{fig:lrw2}
	\end{center}
\end{figure}

\subsection{Structure and obstructions} We will now discuss about
some consequences of Theorem \ref{thm:lrw1}, particularly the
structure of graphs of linear rank-width $1$. 

Let $G$ be a graph and let $xy$ be an edge of $G$. The pivoting of $G$
on $xy$ is the graph $G*x*y*x=G*y*x*y$ \cite{Bouchet88,Oum05}. A graph
$H$ is a \emph{pivot-minor} of a graph $G$ if $H$ can be obtained from
$G$ by a sequence of pivotings and deletions of vertices. It is clear
that a pivot-minor is also a vertex-minor. A graph $H$ is a
vertex-minor (or pivot-minor or induced subgraph) obstruction for
linear rank-width $1$ if $H$ has linear rank-width $2$ and every
proper vertex-minor (or pivot-minor or induced subgraph) of $H$ has
linear rank-width $1$. In \cite{AdlerFP13} the authors gave the
induced subgraph (or vertex-minor or pivot-minor) obstructions for
linear rank-width $1$. We will explain how to obtain all these
obstructions in a more direct way from Theorem \ref{thm:lrw1}.  The
proof of the following is straightforward from Theorem \ref{thm:lrw1}.

\begin{prop}\label{prop:obstructions} A graph $G$ is a distance
  hereditary graph obstruction for linear rank-width $1$ if and only
  if its split decomposition tree is a star with three leaves and
  every connected component of every proper induced subgraph of $G$
  has a path as split decomposition tree.
\end{prop}

From Theorem \ref{thm:lrw1}, we can also deduce that if a connected
graph has linear rank-width $1$, then for each internal node $u$ of
$T_G$, $V(u) \ne \emptyset$ (recall that $V(u)$ is the set of vertices
of $G$ in the block $b_G(u)$). Bouchet \cite{Bouchet88} characterises
exactly distance hereditary graphs such that each internal node of the
split decomposition tree contains at least one vertex.

\begin{thm}[\cite{Bouchet88}]\label{thm:lctree}
 A distance hereditary graph $G$ is locally equivalent to a tree if
 and only if, for every internal node $u$ of $T_G$, $V(u)\ne
 \emptyset$. This tree is moreover unique.
\end{thm}

In the following we deduce some interesting characterisations of
graphs of linear rank-width $1$ we can deduce from Theorem
\ref{thm:lrw1} and some results in the literature. 

\begin{cor}
\label{cor:locCaterpillar} Let $G$ be a connected graph. The following
statements are equivalent. 
\begin{enumerate}
\item $G$ is distance hereditary and $T_G$ is a path.
\item $G$ has linear rank-width $1$.
\item $G$ is locally equivalent to a caterpillar.
\item $G$ is a vertex-minor of a path.
\item $G$ does not contain neither the co-$(3K_2)$ graph nor the Net
  graph nor the $5$-cycle as a vertex-minor. 
\end{enumerate}
\end{cor}

\begin{proof} The equivalence between (1) and (2) is from Theorem
  \ref{thm:lrw1}. The equivalence between (2) and (4) is proved in
  \cite{KwonO13}, but can be easily proved from the equivalence
  between (2) and (3).  Indeed, connected vertex-minors of paths have
  linear rank-width $1$ and since every caterpillar is a vertex-minor
  of a path, we are done. \medskip

  The equivalence between (2) and (5) is proved in
  \cite{AdlerFP13}. Let us explain how to prove it in a more direct
  way from Proposition \ref{prop:obstructions}.  It is sufficient to
  construct the set of induced subgraph obstructions since the set of
  vertex-minor and pivot-minor obstructions can be derived from that
  set. The set of induced subgraph obstructions to being a distance
  hereditary graph is known for a while \cite{BandeltM86} and
  constitutes a subset of the induced subgraph obstruction for linear
  rank-width $1$. We now explain how to construct the set of distance
  hereditary induced subgraph obstructions for linear rank-width
  $1$. Let $H$ be a distance hereditary obstruction for linear
  rank-width $1$. From Proposition \ref{prop:obstructions} its split
  decomposition tree $T_H$ is a star with three leaves, say centred at
  $v$ with $v_1,v_2$ and $v_3$ as leaves. The following three cases
  describe exactly the canonical split decomposition of $H$.  \medskip

  \noindent \textbf{Case 1: $b_H(v)$ is a clique.} Then $b_H(v)$ has
  exactly three vertices and none of the
  $b_H(v_i)$s is a clique. Moreover, for each $1\leq i\leq 3$ the graph
  $b_H(v_i)$ has three vertices and is a star. \medskip

  \noindent \textbf{Case 2: $b_H(v)$ is a star centred at a marker
    vertex.} Then $b_H(v)$ has three vertices. Let us assume without
  loss of generality that the centre of $v$ is the neighbour marker of
  the marker vertex in $v_1$. Then $b_H(v_1)$ is either a clique or a
  star centred at a marker vertex, and for each $2\leq i \leq 3$ the
  graph $b_H(v_i)$ has three vertices and is either a clique or a star
  centred at a vertex of $H$. \medskip

  \noindent \textbf{Case 3: $b_H(v)$ is a star centred at a vertex of
    $H$.} Then $b_H(v)$ has four vertices and its centre is a vertex
  of $H$. Moreover, for each $1\leq i \leq 3$ the graph $b_H(v_i)$ has
  three vertices and is either a clique or a star centred at a vertex
  of $H$. \medskip
  
  From the description of the split decomposition of an
  induced subgraph obstruction for linear rank-width $1$, one can
  clearly construct all the induced subgraph obstructions for linear
  rank-width $1$.  In Figure \ref{fig:lrw1-obs} we have recalled the
  vertex-minor obstructions for linear rank-width $1$. See
  \cite{AdlerFP13} for the complete list of pivot-minor and induced
  subgraph obstructions for linear rank-width $1$.\medskip

  It remains now to prove the equivalence between (2) and (3).
  Caterpillars have clearly linear rank-width $1$ and are the only
  trees with paths as split decomposition trees. Now if a graph $G$
  has linear rank-width $1$, then from Theorem \ref{thm:lrw1} it is a
  distance hereditary graph and its split decomposition tree $T_G$ is
  a path and each node $u$ is such that $V(u)\ne \emptyset$. By
  Theorem \ref{thm:lctree} it is then locally equivalent to a tree,
  say $T$. By \cite[Theorem 4.4]{Bouchet88} $T_G$ is also the split
  decomposition tree of $T$, which concludes the proof.
\end{proof}

\begin{figure}[h!]
 	\begin{center}
		\includegraphics[scale=1.3]{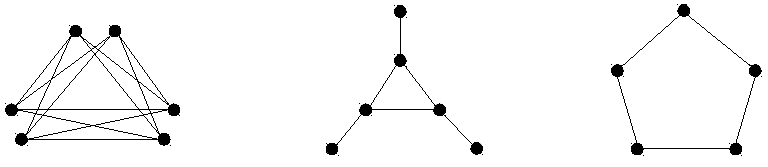}
	\caption{Forbidden vertex-minors for linear rank-width $1$ graphs.
a) co-($3K_2$) b) Net and c) $C_5$.  }
	\label{fig:lrw1-obs}
	\end{center}
\end{figure}

\subsection{Recognition algorithm}
Thanks to the previous characterisations, graphs of linear rank-width
$1$ can be recognised in time $O(n+m)$.

\begin{thm}\label{thm:lrw1-recognition} One can decide in time
  $O(n+m)$ if a connected graph $G$ with $n$ vertices and $m$ edges have linear
  rank-width $1$, and if not construct an obstruction. A linear
  ordering of cutrank $1$ can be constructed also in time $O(n+m)$ if
  it exists. 
\end{thm}

\begin{proof} Let $G$ be a graph with $n$ vertices and $m$
  edges. Thanks to \cite{DamiandHP01} one can check in time $O(n+m)$
  if $G$ is a distance hereditary graph and if not exhibit an induced
  subgraph obstruction. From the induced subgraph obstruction one can
  exhibit, if he prefers, vertex-minor or pivot-minor obstructions.

  We will now assume that $G$ is distance hereditary. We first
  construct a split decomposition tree $T_G$ of $G$ which can be done
  in time $O(n+m)$ \cite{Dahlhaus00}. By Theorem
  \ref{thm:lrw1} $G$ has linear rank-width $1$ if and only if $T_G$ is
  a path. Since testing whether $T_G$ is a path can be done in time
  $O(|V(T_G)|)=O(n)$, we can test in time $O(n+m)$ if $G$ has linear
  rank-width at most $1$. Knowing that $T_G$ is a path, one can
  construct a linear ordering of $V(G)$ of cutrank $1$ in time $O(n)$
  by Proposition \ref{prop:sdtlrw1}.

  We now explain how to exhibit an induced subgraph obstruction if
  $T_G$ is not a path since from an induced subgraph obstruction one
  can exhibit a vertex-minor or a pivot-minor obstruction. If $T_G$ is
  not a path, then there exists an internal node $v$ of degree at
  least three, that can be found in time $O(n)$, and let us choose three
  of its neighbour nodes say $v_1,v_2,v_3$. We need to look at the
  type of the node $v$. \medskip

  \noindent \textbf{Case 1: $b_G(v)$ is a clique.} Then none of the
  $v_i$s is a clique. In each of the graphs $G^{v_iv}$ take either
  two non adjacent vertices that are adjacent to vertices in
  $V(G)\setminus V(G^{v_iv})$ or two adjacent vertices such that
  exactly one is adjacent to vertices in $V(G)\setminus
  V(G^{v_iv})$. \medskip

  \noindent \textbf{Case 2: $b_G(v)$ is a star centred at a marker
    vertex.} We may assume without loss of generality in this case
  that there exists a marker vertex in $v_1$ which is a neighbour
  marker of the centre of $b_G(v)$. Choose in $V(G)\setminus
  V(G^{v_1v})$ either two adjacent vertices or two non adjacent
  vertices, and in each of the graphs $G^{v_iv}$, for $2\leq i \leq 3$
  choose two adjacent vertices such that at least one is adjacent to a
  vertex in $V(G)\setminus V(G^{v_iv})$. \medskip

\noindent \textbf{Case 3: $b_G(v)$ is a star centred at a vertex of
  $G$.} Take the centre of $b_G(v)$, and choose in each of the graphs
$G^{v_iv}$ two adjacent vertices such that at least one is adjacent to
a vertex in $V(G)\setminus V(G^{v_iv})$. \medskip

  \noindent One checks easily that in each of the cases above the split
  decomposition tree of the chosen induced subgraph is a star with three
  leaves and is minimal with respect to that property, hence is an
  induced subgraph obstruction. 
\end{proof}

\bibliography{lrw1}
\bibliographystyle{amsplain}

\today

\end{document}